\documentclass[3p, review]{elsarticle}
\usepackage{amsmath}
\usepackage{amsthm}
\usepackage{amssymb}

\usepackage{subfig}
\usepackage{graphicx}
\usepackage{booktabs}
\usepackage{braket}
\usepackage{enumerate}
\usepackage[title]{appendix}
\usepackage[colorlinks=true, allcolors=blue]{hyperref}
\usepackage{pgfplots}
\newtheorem{theorem}{Theorem}

\newtheorem{proposition}{Proposition}

\theoremstyle{definition}
\newtheorem{defi}{Definition}

\usepackage{diagbox}
\usepackage{makecell}
\usepackage[most]{tcolorbox}
\usepackage{booktabs}
\usepackage{makecell}
\newtcbtheorem{mytcbprob}{Problem}{}{problem}
\usepackage{bm}
\usepackage{multirow}
\usepackage{algorithm}
\usepackage{algpseudocode}

\allowdisplaybreaks[4]

\title{Improved distributed quantum algorithm for  Simon's problem}
\begin{document}
\begin{frontmatter}
\author{Hao Li $^{\rm a,b}$}
\author{Daowen Qiu $^{\rm a,b}$}
\address{$^{\rm a}$ Institute of Quantum Computing and Software, School of Computer Science and Engineering, \\ Sun Yat-sen University, Guangzhou 510006, China;}
\address{$^{\rm b}$The Guangdong Key Laboratory of Information Security Technology, \\Sun Yat-sen University, Guangzhou 510006, China;}
\cortext[mycorrespondingauthor]{issqdw@mail.sysu.edu.cn (D. Qiu).}

\begin{abstract}
Simon’s problem is one of the most important problems demonstrating the power of quantum computing.  Recently, an interesting  distributed quantum algorithm for  Simon's problem was proposed, where  a key sorting operator requiring a large number of qubits was employed. 
In this paper, we design an improved distributed quantum algorithm for  Simon's problem without using sorting operators, and  our algorithm has the advantage of reducing half number of qubits required for a single computing node. Moreover, our algorithm does not involve the classical search process.
\end{abstract}
\end{frontmatter}

\section{Introduction}{\label{Sec1}}

Quantum computing has demonstrated significant potential in solving critical problems \cite{nielsen_quantum_2010}. However, the development of large-scale universal  quantum computers is still slow due to the limitations of the current state of the art. Therefore, in the Noisy Intermediate-Scale Quantum (NISQ) era \cite{preskill_quantum_2018}, 
it is  fascinating to design new quantum algorithms  that require fewer qubits compared with previous quantum algorithms for solving the same problems. 

Distributed quantum computing (DQC) is a novel and attractive area of research, as it aims to solve problems using multiple smaller scale quantum computers. 
DQC has been studied from various approaches and perspectives  (e.g. \cite{avron_quantum_2021}-\cite{Li2025} and the references therein), which can reduce the size and depth of circuits compared to centralized quantum computing, and reduce the depth of circuits and circuit noise to some extent. In the current NISQ era, the use of distributed quantum computing may be beneficial for the implementation of quantum algorithms.

% \cite{avron_quantum_2021,beals_efficient_2013,Hao2023DDJ,Hao2023DGSP,Qiu2017DQC,Qiu24,Tan2022DQCSimon,Xiao2023DQAShor,Xiao2023DQAkShor, Hao2024DMA}

Simon's problem is one of the most important problems in  quantum computing \cite{simon_power_1997}, which is a special kind of the hidden subgroup problem \cite{kaye_introduction_2007}. For solving Simon's problem, quantum algorithms have the advantage of exponential speedup over  classical algorithms \cite{cai_optimal_2018}. In a way, Simon's algorithm inspired the  proposal of Shor's algorithm  \cite{nielsen_quantum_2010}. 
An exact quantum algorithm for  Simon’s problem  with $O(n)$ queries was proposed  as well \cite{cai_optimal_2018}.   Then, a classical deterministic algorithm for
Simon’s problem with $O(\sqrt{2^n})$ queries is designed \cite{cai_optimal_2018}, and therefore the optimal separation of the quantum exact query complexity and the classical exact query complexity of Simon's problem is $\Theta(n)$ versus $\Theta(\sqrt{2^n})$.

Recently,  an interesting distributed quantum algorithm for Simon’s problem   was proposed \cite{Tan2022DQCSimon}. 
However, the sorting operator in the algorithm of  \cite{Tan2022DQCSimon} requires a  large number of qubits.
 Moreover, the algorithm in \cite{Tan2022DQCSimon} requires to split the hidden string $s$ to be found into two parts, $s_1$ and $s_2$, by finding $s_1$  and $s_2$ one after another.

In this paper, we design an improved distributed quantum algorithm for Simon's problem. Essentially, our algorithm replaces the sorting operator in   \cite{Tan2022DQCSimon} with a new unitary operator. In fact, our algorithm achieves the consistent  effect with Simon's algorithm. Specifically, our algorithm has the following advantages:
\begin{itemize}
\item 
Compared to the algorithm in \cite{Tan2022DQCSimon},  our algorithm requires fewer qubits.
\item 
Different from the algorithm of \cite{Tan2022DQCSimon}  for finding the hidden string $s$ by two parts, i.e. quantum algorithm and classical algorithm, our algorithm can obtain $s$ directly.
\end{itemize}

The remainder of this paper is organized as follows. In Sec. \ref{Sec2}, we  introduce  Simon's problem in  distributed scenario and recall the distributed quantum algorithm for Simon’s problem \cite{Tan2022DQCSimon}. 
Then, in Sec. \ref{Sec3}, we give an improved  distributed   quantum algorithm for  Simon's problem  and the corresponding analytical procedure.   Afterwards, in Sec. \ref{Sec4}, we compare our algorithm with Simon's algorithm and the algorithm in \cite{Tan2022DQCSimon}. Finally,  we  conclude with a summary in Sec. \ref{Sec5}.

\section{Preliminaries}\label{Sec2}

In this section, we introduce Simon's problem and  Simon's problem in  distributed scenario, and recall the distributed quantum algorithm for Simon’s in \cite{Tan2022DQCSimon}. 
Below,  we first present some notations.

For $x, y\in\mathbb{Z}_2^n$ with $x=(x_1,\ldots,x_n)$ and $y=(y_1,\ldots,y_n)$,  denote
\begin{align}
x\oplus y=((x_1+y_1)\bmod2,\ldots, (x_n+y_n)\bmod2).
\end{align}
\begin{align}
x\cdot y=(x_1\cdot y_1+\cdots +x_n\cdot y_n)\bmod2.
\end{align}

For any bit string $s\in \{0, 1\}^n$,  denote
\begin{align}
s^{\perp}=\{z \in \{0, 1\}^n|s\cdot z = 0\}.
\end{align}

In the following, we introduce Simon's problem.

\begin{mytcbprob*}{Simon's problem}
\textbf{Input:} A  function $f:\{0,1\}^n \rightarrow \{0,1\}^m$, where $m\geq n-1$.

\textbf{Promise:} There exists a hidden string $s\in\{0,1\}^n$ such that for any $x, y\in {\{0, 1\}}^n$, $f(x) = f(y)$ iff $x \oplus y \in s$.

\textbf{Output:} The hidden string $s$.
\end{mytcbprob*}

Next, we  describe Simon's problem in  distributed scenario.
\begin{mytcbprob*}{Simon's problem in  distributed scenario}
\textbf{Input:} A  function $f:\{0,1\}^n \rightarrow \{0,1\}^m$, which is divided into $2^t$ subfunctions $f_w:\{0,1\}^{n-t}\rightarrow\{0,1\}^m$ as $f_w(u)=f(uw)$, where $m\geq n-1$, $1\leq t<n$, $u \in \{0,1\}^{n-t}$ and $w\in\{0,1\}^t$.

\textbf{Promise:} There exists a hidden string $s\in\{0,1\}^n$ such that for any $x, y\in {\{0, 1\}}^n$, $f(x) = f(y)$ iff $x \oplus y \in s$.

\textbf{Output:} The hidden string $s$ by querying the $2^t$ subfunctions $f_w$.
\end{mytcbprob*}

 In the following, we  briefly review the distributed quantum algorithm for Simon’s problem \cite{Tan2022DQCSimon}.
First, we  introduce the notations and operators used in  \cite{Tan2022DQCSimon}.
 For any operator $A_w$ with $w\in \{0,1\}^t$, denote
$\prod\nolimits_{w\in\{0,1\}^t}A_w\triangleq A_{1^t}A_{1^{t-1}0}\cdots$ $A_{0^t}$, and 
$\prod'\nolimits_{w\in\{0,1\}^t}A_w\triangleq A_{0^t}A_{0^{t-1}1}\cdots A_{1^t}$.

Let $[N_t]$ represent the set of integers $\{0,1,\cdots, 2^t-1\}$. Let ${\rm BI}:\{0,1\}^t \rightarrow [N_t]$ be the function to convert a bit string of $t$ bits to an equal decimal integer.
The query operators $O^*_{f_w}$ in Algorithm \ref{algorithm1}
are defined as 
\begin{equation}
O^*_{f_w}\ket{u}\ket{b}\ket{c}=\ket{u}\ket{b}\ket{c\oplus f_w(u)},
\end{equation}
where $u\in\{0,1\}^{n-t}$, $w\in \{0,1\}^t$, $b\in\{0,1\}^{t+{\rm BI}(w)\cdot m}$ and $c\in \{0,1\}^m$. 

The active qubits for oracle $O^*_{f_w}$ are defined as control qubits $u$ and target qubits $c$, so the number of active qubits for oracle $O^*_{f_w}$ is $n-t+m$. 
 
 \begin{defi}
  For any $u \in \{0,1\}^{n-t}$, let 
  \begin{align}
  S(u)=f_{w_0}(u)f_{w_1}(u)\cdots f_{w_{2^t-1}}(u),
  \end{align}
  where $f_{w_0}(u)\preccurlyeq f_{w_1}(u)\preccurlyeq \ldots \preccurlyeq f_{w_{2^t-1}}(u)\in \{0,1\}^m$, with $w_i\in\{0,1\}^t$ $(0\leq i\leq 2^t-1)$, where $w_i\neq w_j$ for any $i\neq j$, and $\preccurlyeq$ denotes the lexicographical order.
\end{defi}
 
 The  {sorting} operator $U_{\mathrm{Sort}}:\{0,1\}^{2^{t+1}m}\rightarrow\{0,1\}^{2^{t+1}m}$ in Algorithm \ref{algorithm1}  
is defined as 
\begin{equation}\label{U_{Sort}}
\begin{split}
U_{\mathrm{Sort}}\left(\bigotimes_{w\in\{0,1\}^{t}}\ket{f_w(u)}\right)|b\rangle
=\left(\bigotimes_{w\in\{0,1\}^{t}}\ket{f_w(u)}\right)\Ket{b\oplus  S(u)},
\end{split}
\end{equation}
where   $\bigotimes\limits_{w\in\{0,1\}^{t}}\ket{f_w(u)}\triangleq \ket{f_{0^t}(u)}\ket{f_{0^{t-1}1}(u)}\cdots \ket{f_{1^t}(u)}$ and $b\in \{0,1\}^{2^tm}$.

Let $s$ be the hidden string to be found, and denote $s=s_1s_2$,
where the length of $s_1$ is $n-t$,  the length of $s_2$ is $t$. The algorithm in \cite{Tan2022DQCSimon} comprises two sub-algorithms: Algorithm \ref{algorithm1}  for finding $s_1$ and  Algorithm \ref{algorithm2}  for finding $s_2$. 
 The following is Algorithm \ref{algorithm1} and its circuit diagram is shown in Fig.  \ref{fig1}.
\begin{algorithm}[H]
\caption{Distributed quantum algorithm for  finding $s_1$}
\label{algorithm1}
\begin{algorithmic}
\State 1: $\Ket{\phi_0}=\Ket{0^{n-t}}\Ket{0^{2^{t+1}m}}$;
\State2: $\ket{\phi_1}=\left(H^{\otimes n-t}\otimes I^{\otimes {2^{t+1}m}}\right)\ket{\phi_0}$;
\State 3: $\ket{\phi_2}=\left(\prod_{w\in\{0,1\}^t}\left(O^*_{f_{w}}\otimes I^{\otimes \left(2^{t+1}-{\rm BI}(w)-1\right)m}\right)\right)\ket{\phi_1}$;
\State 4: $\ket{\phi_3}=\left(I^{\otimes {n-t}}\otimes U_{\mathrm{Sort}} \right)\ket{\phi_2}$;
\State 5: $\ket{\phi_4}=\left(\prod\nolimits_{w\in\{0,1\}^t}'\left(O^*_{f_{w}}\otimes I^{\otimes \left(2^{t+1}-{\rm BI}(w)-1\right)m}\right)\right)\ket{\phi_3}$;
\State 6: $\ket{\phi_5}=\left(H^{\otimes n-t}\otimes I^{\otimes {2^{t+1}m}}\right)\ket{\phi_4}$;
\State 7: Measure the first $n-t$ qubits of $\ket{\phi_5}$ and get an element in $s_1^{\perp}$.
\end{algorithmic}
\end{algorithm}

\begin{figure}[H]
  \includegraphics[width=\textwidth]{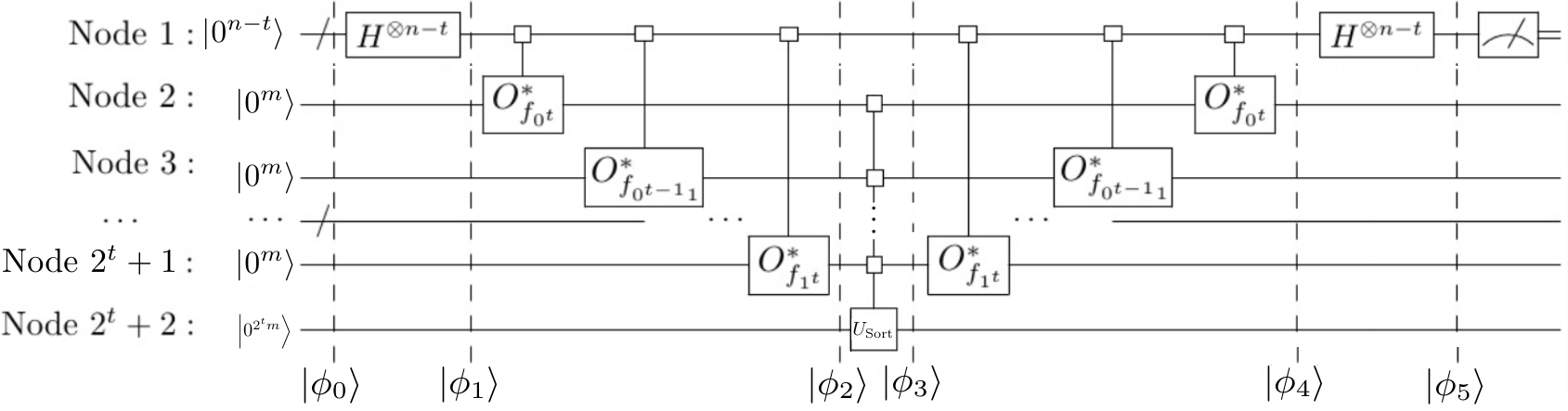}
  \caption{The circuit for  the quantum part of the distributed quantum algorithm for finding $s_1$.}
  \label{fig1}
\end{figure}

After $O(n-t)$ repetitions of Algorithm \ref{algorithm1},
it is possible to get $O(n- t)$ elements in $s_1^{\perp}$. Then, by using the classical Gaussian elimination method, $s_1$ is obtained. If $s_1$ has already been found, $s_2$ can be found in terms of  Algorithm \ref{algorithm2}.
 Finally, $s=s_1s_2$ can be obtained. The following is Algorithm \ref{algorithm2}.
 \begin{algorithm}[H]
\caption{Distributed quantum algorithm for  finding $s_2$}
\label{algorithm2}
\begin{algorithmic}
\State 1: Query each oracle $O^*_{f_w}$ once in parallel to get
$f(0^{n-t}w)$ $(w\in\{0,1\}^t)$;
\State 2: Query oracle $O^*_{f_0^t}$ once to get $f(s_10^t)$;
\State 3: Find a $v\in\{0,1\}^t$ such that  $f(0^{n-t}v) = f(s_10^t)$;
\State 4: Obtain $s_2 = v$.
\end{algorithmic}
\end{algorithm}

\section{Improved distributed quantum algorithm for  Simon's problem}\label{Sec3}

In this section, we describe an improved distributed quantum algorithm for  Simon's problem, i.e., Algorithm \ref{algorithm3}.

The main design idea of our algorithm is to simplify the algorithm  in \cite{Tan2022DQCSimon}. Specifically, we replace the sorting operator $U_{\mathrm{Sort}}$ in the algorithm in \cite{Tan2022DQCSimon} with the operator $V$ that requires fewer qubits. Essentially, the effect of our algorithm is the same as Simon's algorithm.

The operator $V:\{0,1\}^{(2^{t}+1)m+t}\rightarrow\{0,1\}^{(2^{t}+1)m+t}$  
is defined as 
\begin{equation}\label{V}
\begin{split}
V\ket{i}\left(\bigotimes_{j\in\{0,1\}^{t}}\ket{a_j}\right)|b\rangle
=\ket{i}\left(\bigotimes_{j\in\{0,1\}^{t}}\ket{a_j}\right)\Ket{b\oplus a_i},
\end{split}
\end{equation}
where  $i\in\{0,1\}^{t}$, $a_j\in \{0,1\}^m$, $b\in \{0,1\}^m$, and
$\bigotimes\nolimits_{j\in\{0,1\}^{t}}\ket{a_j}\triangleq \ket{a_{0^t}}\ket{a_{0^{t-1}1}}\cdots $ $\ket{a_{1^t}}$.

In order to easily visualize the construction of the operator $V$ in our algorithm, we give an example below.
Let  $t=1$ and $m=2$, then the circuit of  $V$ is shown in Fig. \ref{fig2}. As shown in Fig. \ref{fig2}, the circuit of  $V$ contains $X$ gates and Toffoli gates.

\begin{figure}[H]
  \centerline{\includegraphics[width=0.8\textwidth]{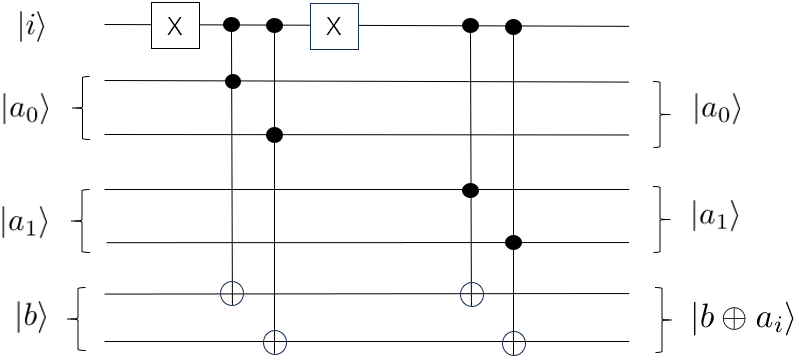}}
  \caption{The circuit for  operator $V$ with parameters $t=1$ and $m=2$.}
    \label{fig2}
\end{figure}

The following is Algorithm \ref{algorithm3} and its circuit diagram is shown in Fig.  \ref{fig3}.
\begin{algorithm}[H]
\caption{Improved distributed quantum algorithm for  Simon's problem}
\label{algorithm3}
\begin{algorithmic}
\State 1: $\Ket{\psi_0}=\Ket{0^{n-t}}\Ket{0^{t}}\Ket{0^{(2^{t}+1)m}}$;
\State2: $\ket{\psi_1}=\left(H^{\otimes n-t}\otimes H^{\otimes t}\otimes I^{\otimes {\left(2^{t}+1\right)m}}\right)\ket{\psi_0}$;
\State 3: $\ket{\psi_2}=\left(\prod_{w\in\{0,1\}^t}\left(O^*_{f_{w}}\otimes I^{\otimes \left(2^t-{\rm BI}(w)\right)m}\right)\right)\ket{\psi_1}$;
\State 4: $\ket{\psi_3}=\left(I^{\otimes {n-t}}\otimes V \right)\ket{\psi_2}$;
\State 5: $\ket{\psi_4}=\left(\prod\nolimits_{w\in\{0,1\}^t}'\left(O^*_{f_{w}}\otimes I^{\otimes \left(2^t-{\rm BI}(w)\right)m}\right)\right)\ket{\psi_3}$;
\State 6: $\ket{\psi_5}=\left(H^{\otimes n-t}\otimes H^{\otimes t}\otimes I^{\otimes {\left(2^{t}+1\right)m}}\right)\ket{\psi_4}$;
\State 7: Measure the first $n$ qubits of $\ket{\psi_5}$ and get an element in $s^{\perp}$.
\end{algorithmic}
\end{algorithm}

\begin{figure}[H]
  \includegraphics[width=\textwidth]{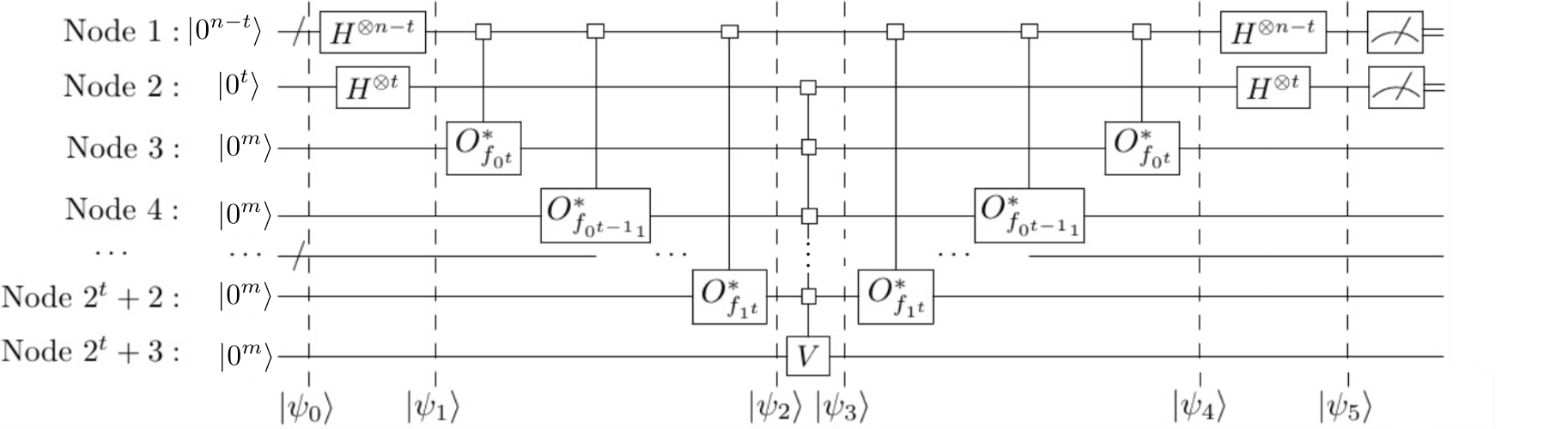}
  \caption{The circuit for  improved distributed quantum algorithm.}
  \label{fig3}
\end{figure}

In the following, we give a theorem related to proving the correctness of our algorithm, where the proof describes the evolution of quantum states in our algorithm.
\begin{theorem}\label{The1} 
In line 7 of Algorithm \ref{algorithm3}, the measuring result $z$ is obtained and satisfies $z\in s^{\perp}$  with certainty.
\end{theorem}

\begin{proof}
The state after the first step  of  Algorithm \ref{algorithm3} is
\begin{align}
  \ket{\psi_1}=&\left(H^{\otimes n-t}\otimes H^{\otimes t}\otimes I^{\otimes {\left(2^{t}+1\right)m}}\right)\ket{\psi_0}\\
  =&\frac{1}{\sqrt{2^{n-t}}}\sum_{u\in\{0,1\}^{n-t}}|u\rangle\frac{1}{\sqrt{2^{t}}}\sum_{w\in\{0,1\}^{t}}|w\rangle\Ket{0^{(2^{t}+1)m}}.
\end{align}

After applying $O^*_{f_{w}}$, the following state is obtained
\begin{align}
  \ket{\psi_2}=&\left(\prod_{w\in\{0,1\}^t}\left(O^*_{f_{w}}\otimes I^{\otimes \left(2^t-{\rm BI}(w)\right)m}\right)\right)\ket{\psi_1}\\
  =&\frac{1}{\sqrt{2^{n-t}}}\sum_{u\in\{0,1\}^{n-t}}\ket{u}\frac{1}{\sqrt{2^{t}}}\sum_{w\in\{0,1\}^{t}}\ket{w}
  \ket{f_{0^t}(u)}\ket{f_{0^{t-1}1}(u)}\cdots\ket{f_{1^t}(u)}\ket{0^m}.
\end{align}

Then acting with $V$, we have
\begin{align}
  |\psi_3\rangle=&\left(I^{\otimes {n-t}}\otimes V \right)\ket{\psi_2}\\=&\frac{1}{\sqrt{2^{n-t}}}\sum_{u\in\{0,1\}^{n-t}}\ket{u}\frac{1}{\sqrt{2^{t}}}\sum_{w\in\{0,1\}^{t}}\ket{w}\ket{f_{0^t}(u)}\ket{f_{0^{t-1}1}(u)}\cdots\ket{f_{1^t}(u)}\ket{f_w(u)}.
\end{align}

Further, by   $O^*_{f_{w}}$ $(w\in\{0,1\}^t)$  the following state yields
\begin{align}
\ket{\psi_4}=&\left(\prod\nolimits_{w\in\{0,1\}^t}'\left(O^*_{f_{w}}\otimes I^{\otimes \left(2^t-{\rm BI}(w)\right)m}\right)\right)\ket{\psi_3}\\
=&\frac{1}{\sqrt{2^{n-t}}}\sum_{u\in\{0,1\}^{n-t}}\ket{u}\frac{1}{\sqrt{2^{t}}}\sum_{w\in\{0,1\}^{t}}\ket{w}\Ket{0^{2^tm}}\ket{f_w(u)}.
\end{align}

Continuing to apply Hadamard transform to the first $n$ qubits of $\ket{\psi_4}$, we obtain the following state
\begin{align}
	|\psi_{5}\rangle=&\left(H^{\otimes n-t}\otimes H^{\otimes t}\otimes I^{\otimes {\left(2^{t}+1\right)m}}\right)\ket{\psi_4}\\
	=&\frac{1}{2^{n-t}}\sum_{u,u'\in\{0,1\}^{n-t}}(-1)^{u\cdot u'}\ket{u'}\frac{1}{2^{t}}\sum_{w,w'\in\{0,1\}^{t}}(-1)^{w\cdot w'}\ket{w'}\Ket{0^{2^tm}}\ket{f_w(u)}\\
    =&\frac{1}{2^{n}}\sum_{\substack{uw,u'w'\in\{0,1\}^{n}}}(-1)^{(uw)\cdot (u'w')}\ket{u'w'}\Ket{0^{2^tm}}\ket{f(uw)}\\
    =&\frac{1}{2^{n}}\sum_{\substack{x,z\in\{0,1\}^{n}}}(-1)^{x\cdot z}\ket{z}\Ket{0^{2^tm}}\ket{f(x)} (\text{denote}\ x=uw, z=u'w').
  \end{align} 

Then, we have
 \begin{align}  
    |\psi_{5}\rangle=&\frac{1}{2^{n+1}}\sum_{\substack{x,z\in\{0,1\}^{n}}}\left((-1)^{x\cdot z}\ket{z}\Ket{0^{2^tm}}\ket{f(x)}
    +(-1)^{(x\oplus s)\cdot z}\ket{z}\Ket{0^{2^tm}}\right.\ket{f(x\oplus s)}\Big)\\
    =&\frac{1}{2^{n+1}}\sum_{\substack{x,z\in\{0,1\}^{n}}}\left((-1)^{x\cdot z}\ket{z}\Ket{0^{2^tm}}\ket{f(x)}
    +\right.(-1)^{(x\oplus s)\cdot z}\ket{z}\Ket{0^{2^tm}}
    \ket{f(x)}\Big) (\text{by}\ f(x)=f(x\oplus s)) \\
    =&\frac{1}{2^{n+1}}\sum_{\substack{x,z\in\{0,1\}^{n}}}(-1)^{x\cdot z}(1+(-1)^{s\cdot z})\ket{z}\Ket{0^{2^tm}}\ket{f(x)}.
\end{align} 

If $s\cdot z =1$, then $1 + (-1)^{s\cdot z}= 0$. Consequently, the basis state $\ket{z}$ vanishes in the above state. If $s\cdot z = 0$, we have $1 + (-1)^{s\cdot z} = 2$. Thus, we have

\begin{align}
	|\psi_{5}\rangle
    =&\frac{1}{2^{n}}\sum_{\substack{z\in s^{\perp}}}\ket{z}\sum_{\substack{x\in\{0,1\}^{n}}}(-1)^{x\cdot z}\Ket{0^{2^tm}}\ket{f(x)}.
\end{align} 

Therefore, in line 7 of Algorithm \ref{algorithm3}, the measuring result $z$  satisfies $z\in s^{\perp}$  with certainty.

\end{proof}

After $O(n)$ repetitions of Algorithm \ref{algorithm3},
 $O(n)$ elements in $s^{\perp}$ are obtained. Then, using the classical Gaussian elimination method, we can obtain $s$.

\section{Comparison of algorithms}\label{Sec4}

 In this section, we compare our algorithm with Simon's algorithm  \cite{simon_power_1997} and the distributed quantum algorithm for  Simon's problem \cite{Tan2022DQCSimon}.

First, we compare our algorithm, i.e., Algorithm 3, with Simon's algorithm.
In Algorithm \ref{algorithm3}, the number of active qubits for each oracle is only $n - t + m$. However, in Simon's algorithm, the
number of active qubits for the oracle is $n + m$ \cite{simon_power_1997}. 

The quantum query complexity of Simon's algorithm is $O(n)$ \cite{simon_power_1997}. The quantum query complexity of  Algorithm ref{algorithm3} can be thought of as the minimum number of queries to query a single sub-oracle. Since Algorithm \ref{algorithm3} need run $O(n)$ times to solve  Simon's problem, we can deduce that its quantum query complexity is $O(n)$. Comparison of Simon's algorithm with Algorithm \ref{algorithm3} is presented in Tab. \ref{tab1}.

\begin{table}[H] 
\renewcommand{\arraystretch}{1.3}
	\centering
	\caption{Comparison of  Simon's algorithm  with Algorithm \ref{algorithm3}.}
%\resizebox{\linewidth}{!}{ 
	\begin{tabular}{*{4}{c}}
		\hline
		                 Algorithms   & \makecell[c]{The
number of active qubits \\ for each oracle}   &\  \makecell[c]{Quantum query  complexity}   \\
		\hline
		  \makecell[c]{Simon's algorithm \cite{simon_power_1997}}     &        $n+m$ & $O(n)$   \\
		  \makecell[c]{Algorithm \ref{algorithm3}}       &      $n+m-t$ & $O(n)$  \\	
		\hline
	\end{tabular}
	%}
  \label{tab1}
\end{table}

Then, 
we compare Algorithm \ref{algorithm1} with Algorithm \ref{algorithm3}.  Comparison of Algorithm \ref{algorithm1} with Algorithm \ref{algorithm3} is shown in Tab. \ref{tab2}. 
In the following, we give a proposition concerning the space complexity of Algorithm \ref{algorithm1} and Algorithm \ref{algorithm3}.

\begin{proposition}
The maximum number of qubits for a single computing node of Algorithm \ref{algorithm1} is $2^{t+1}m$. The maximum number of qubits for a single computing node of Algorithm \ref{algorithm3} is $(2^{t}+1)m+t$.
\end{proposition}
\begin{proof}
In Algorithm \ref{algorithm1}, implementing the operator $U_{\mathrm{Sort}}$ requires uniting Node 2, Node 3, and up
to Node $2^t + 2$. The $m$ qubits from Node 2 are transmitted to Node $2^t+2$. The $m$ qubits  from Node 3 are transmitted to Node $2^t+2$, and so on. Lastly,  the $m$ qubits from Node $2^t+1$      are transmitted to Node $2^t+2$. Thus, the number of qubits required for Node $2^t+2$  implementing the operator $U_{\mathrm{Sort}}$ is $2^{t+1}m$. 
The number of qubits required for the node implementing the query operator $O^*_{f_w}$ is $n - t + m$. Since $m\geq n-1$, it follows that $\max\left\{n - t + m,2^{t+1}m\right\}=2^{t+1}m$. Thus, the maximum number of
qubits  for a single computing node in Algorithm \ref{algorithm1} is $2^{t+1}m$. 

In Algorithm \ref{algorithm3}, implementing the operator $V$ requires uniting Node 2, Node 3, and up
to Node $2^t + 3$. The $t$ qubits from Node 2 are transmitted to Node $2^t+3$. The $m$ qubits from Node 3 are transmitted to Node $2^t+3$. The 
 $m$ qubits from Node 4 are transmitted to Node $2^t+3$, and so on. Lastly,  the $m$ qubits from Node $2^t+2$  are transmitted to Node $2^t+3$. Thus,  the number of qubits required for Node $2^t+3$  implementing the operator $V$ is $(2^{t}+1)m+t$. 
Therefore, it can be deduced that the maximum number of qubits  for a single computing node in Algorithm \ref{algorithm3} is $(2^{t}+1)m+t$. Since both $m$ and $t$ are positive integers,  the maximum number of qubits for a single computing node in Algorithm \ref{algorithm3} is less than that of Algorithm \ref{algorithm1}. 
\end{proof}

In the following, we give a proposition describing the quantum communication complexity of  Algorithm \ref{algorithm1} and Algorithm  \ref{algorithm3}. 

\begin{proposition}
The quantum communication complexity of Algorithm \ref{algorithm1} is $O((n-t)(2^t(n+m-t)))$. The quantum communication complexity of Algorithm \ref{algorithm3} is $O(n(2^t(n+m-t)+t))$. 
\end{proposition}
\begin{proof}
In Algorithm \ref{algorithm1},  after   $\ket{\phi_1}$, the $n-t$ qubits of  Node 1  are transmitted to   $O^*_{f_{0^t}}$. The $n-t$ control qubits of  $O^*_{f_{0^t}}$  are transmitted to $O^*_{f_{0^{t-1}1}}$, and so on. Lastly, the $n-t$ control qubits of  $O^*_{f_{1^{t-1}0}}$  are transmitted to  $O^*_{f_{1^ t}}$. The quantum communication complexity of the process is $O\left(2^t(n-t)\right)$. 
After $\ket{\phi_2}$, the $m$ target qubits of   $O^*_{f_{0^{t}}}$  are transmitted to   $U_{\mathrm{Sort}}$. The $m$ target qubits of   $O^*_{f_{0^{t-1}1}}$  are transmitted to   $U_{\mathrm{Sort}}$,  and so on. Eventually,  the $m$ target qubits  of  $O^*_{f_{1^ t}}$  are transmitted to  $U_{\mathrm{Sort}}$.  The quantum communication complexity of the process is $O \left(2^tm\right)$. 
Finally, the quantum communication complexity of Algorithm \ref{algorithm1} is $O(2^t(n+m-t))$ for one iteration run. Since Algorithm \ref{algorithm1} needs to be run $O(n-t)$ times to find  $s_1$, the quantum communication complexity of Algorithm \ref{algorithm1} is $O((n-t)(2^t(n+m-t)))$. 

In Algorithm \ref{algorithm3}, after  $\ket{\psi_1}$, the $n-t$ qubits of  Node 1  are transmitted to   $O^*_{f_{0^t}}$. The $n-t$ control qubits of  $O^*_{f_{0^t}}$  are transmitted to $O^*_{f_{0^{t-1}1}}$, and so on. Lastly, the $n-t$ control qubits of  $O^*_{f_{1^{t-1}0}}$  are transmitted to  $O^*_{f_{1^ t}}$.  The quantum communication complexity of the process is $O\left(2^t(n-t)\right)$. 
After  $\ket{\psi_2}$, the $t$ qubits of   Node 2  are transmitted to   $V$. The $m$ target qubits of   $O^*_{f_{0^{t}}}$  are transmitted to   $V$. The $m$ target qubits of   $O^*_{f_{0^{t-1}1}}$  are transmitted to  $V$, and so on. Eventually,  the $m$ target qubits  of  $O^*_{f_{1^ t}}$  are transmitted to  $V$.  The quantum communication complexity of the process is $O \left(2^tm+t\right)$. 
Finally, the quantum communication complexity of Algorithm \ref{algorithm3} is $O(2^t(n+m-t)+t)$ for one iteration run. Since Algorithm \ref{algorithm3} needs to be run $O(n)$ times to find  $s$, the quantum communication complexity of Algorithm \ref{algorithm3} is $O(n(2^t(n+m-t)+t))$. 
\end{proof}

\begin{table}[H] 
\renewcommand{\arraystretch}{1.3}
	\centering
	\caption{Comparison of  Algorithm \ref{algorithm1} with Algorithm \ref{algorithm3}.}
%\resizebox{\linewidth}{!}{ 
	\begin{tabular}{*{4}{c}}
		\hline
		                 Algorithms      &  \makecell[c]{Maximum number of
qubits \\ for a
single computing node} &\  \makecell[c]{Quantum communication complexity} \\
		\hline
		  \makecell[c]{Algorithm \ref{algorithm1} \cite{Tan2022DQCSimon}}     &        $2^{t+1}m$ & $O((n-t)(2^t(n+m-t)))$\\
		  \makecell[c]{Algorithm \ref{algorithm3}}       &      $(2^t+1)m+t$ & $O(n(2^t(n+m-t)+t))$  \\	
		\hline
	\end{tabular}
	%}
  \label{tab2}
\end{table}

\section{Conclusion}\label{Sec5}

Inspired by the ideas and methods of  algorithms  in  \cite{Qiu24,Tan2022DQCSimon} and Simon's algorithm \cite{simon_power_1997},
 we have proposed an improved distributed  quantum  algorithm for Simon's problem.
 %Given a function $f:\{0,1\}^n \rightarrow \{0,1\}^m$, it is promised that there exists a hidden string $s\in\{0,1\}^n$ such that for any $x, y\in {\{0, 1\}}^n$, $f(x) = f(y)$ iff $x \oplus y \in s$, where $m\geq n-1$.
%In order to design an improved distributed  quantum  algorithm to  
%find the  $s$,  
Specifically,  we  have designed a distributed query  operator.  
First, we have decomposed  the function in Simon's problem into a number of smaller subfunctions. Then, we have united the query operators of the several subfunctions by using a specific unitary operator.

%$f$ into $2^t$ sub-functions $f_w:\{0,1\}^{n-t}\rightarrow\{0,1\}^m$, where $w\in\{0,1\}^t$. Then, we have united  $2^t$ query operators $O^*_{f_w}$ by employing  the query  operator $V$.

Compared to the  distributed quantum algorithm for Simon’s problem  \cite{Tan2022DQCSimon}, our algorithm  has shown certain advantages. 
The maximum number of qubits for a single computing node in our algorithm is almost half of that of the algorithm in \cite{Tan2022DQCSimon}, which  may accelerate its implementation in the NISQ era.  
In  future, we may design an improved distributed exact quantum algorithm for generalized Simon's problem, by utilizing the technical methods proposed in \cite{Qiu24,Tan2022DQCSimon}  and this paper, as well as the quantum amplitude amplification algorithm.

%\section*{CRediT authorship contribution statement}
 
% \textbf{Hao Li:}  Formal analysis,Validation, Investigation,Writing–original draft.
 % \textbf{Daowen Qiu:} Methodology, Conceptualization, Formal analysis, Writing–original draft, Writing–review \& editing,
 %Supervision, Project administration.

%\section*{Declaration of competing interes}

%The authors declare that they have no known competing financial interests or personal relationships that could have appeared to influence the work reported in this paper.

%\section*{Data availability}

%No data was used for the research described in the article.
 
%\section*{Acknowledgements}
%This work is supported in part by  the Shenzhen Science and Technology Program (Grant No. JCYJ20220818102003006), and the Guangdong Provincial Quantum Science Strategic Initiative (No. GDZX2200001).

%\bibliographystyle{elsarticle-num} 
%\bibliography{ref}

\end{document}